\documentclass[letterpaper]{article}

\usepackage{amsmath}
\usepackage{amsfonts}
\usepackage{amssymb}
\usepackage{amsthm}
\usepackage{graphicx}
\usepackage{moreverb}

\usepackage[dvips,colorlinks,bookmarksopen,bookmarksnumbered,citecolor=red,urlcolor=red]{hyperref}
\usepackage{subfig}

\newcommand\BibTeX{{\rmfamily B\kern-.05em \textsc{i\kern-.025em b}\kern-.08em
T\kern-.1667em\lower.7ex\hbox{E}\kern-.125emX}}

\numberwithin{equation}{section}
\newtheorem{theorem}{Theorem}
\newtheorem{lemma}{Lemma}
\newtheorem{proposition}{Proposition}
\newtheorem{remark}{Remark}
\newtheorem{example}{Example}

\newcommand\RE{\mathbb{R}}
\newcommand\ACal{\mathcal{A}}
\newcommand\SCal{\mathcal{S}}

\DeclareMathOperator{\Convex}{conv} 
\DeclareMathOperator{\Sign}{sign}
\DeclareMathOperator{\Diag}{diag}

\DeclareMathOperator*{\Argmax}{arg\,max}
\DeclareMathOperator*{\Argmin}{arg\,min}

\begin{document}


\title{Min-max piecewise constant optimal control for \\ multi-model linear systems}

\author{F\'elix A. Miranda, Fernando Casta\~nos and Alexander Poznyak}



\maketitle

\begin{abstract}
The present work addresses a finite-horizon linear-quadratic optimal control problem
for uncertain systems driven by piecewise constant controls. The precise values of the
system parameters are unknown, but assumed to belong to a finite set (i.e., there exist
only finitely many possible models for the plant). Uncertainty is dealt with using a min-max
approach (i.e., we seek the best control for the worst possible plant). The optimal control
is derived using a multi-model version of Lagrange's multipliers method, which specifies the
control in terms of a discrete-time Riccati equation and an optimization problem over a simplex.
A numerical algorithm for computing the optimal control is proposed and tested by simulation.
\end{abstract}




\section{Introduction}



Multi-model dynamical systems arise in several areas of control: Some times as
a result of the uncertainty in the system parameters and some times when
the original model is largely complicated and needs to be divided into subsystems,
each of which characterizes an important feature in some region of the parameter
or the state space.
 
For this class of models, the optimal control problem can be formulated in such a
way that an operation of maximization is taken over the set of uncertainties and
an operation of minimization is taken over the control strategies. This is known
as a min-max optimal control problem~\cite{Boltyanski2011}. In this approach, 
the original system model is replaced by a finite set of dynamic models 
such that each model describes a particular uncertain case. 

The multi-model min-max optimal control problem has been considered 
in several works like~\cite{Varga1996,Poznyak2002,Azhmyakov2010,
Besselmann2012, Ramirez2002, Bemporad2003}, to name a few. The purpose
in these references is to obtain a control signal $u(\cdot)$ which guarantees that the cost
does not exceed the ``worst cost'' incurred by the plant realizing the ``worst parameters''.
This problem was solved in~\cite{Poznyak2002} using the so-called robust maximum principle.
On the other hand, the multi-model control problem was studied in~\cite{Azhmyakov2010} from a
dynamic-programming perspective, and a natural relationship between dynamic programming and the
robust maximum principle was established for a class of linear-quadratic (LQ) problems.

In references \cite{Besselmann2012} and \cite{Bemporad2003} a min-max
control problem is solved using a Model Predictive Control (MPC) approach, where the cost functional is considered with a polyhedral
norm (a $p$-norm with $p = 1$ or $\infty$). As reported in 
\cite{Besselmann2012}, there is no efficient technique
for finding the solution of quadratic cost problems ($p = 2$).
The case of norm $p = 2$ is treated in \cite{Ramirez2002} following a
MPC view point for the case of unidimensional control ($m = 1$),
where the authors obtain a characterization of the optimal control in terms of the possible location of the optimal solution, resulting in a
combinatorial problem whose size grows exponentially with the number of plants.

The multi-model control problem is also relevant when the designer
is not only concerned about optimizing control effort, but it is interested in optimizing communication
bandwidth as well. Bandwidth optimization is common in applications belonging to the field of networked
control systems, where the controller is not completely dedicated to the plant or it does not have full
access to the network resources at every time (see e.g.~\cite{Hespanha2007,Yang2006}). The communication
constraint leads to the consideration of a particular set of admissible controls, given by piecewise
constant functions on the interval $[t_0, t_N]$, which in general are not uniformly spaced in time.
Non-uniformity is motivated by the fact that, in terms of bandwidth optimization, a uniform
sampling is not necessarily a ``good'' choice~\cite{Bini2014}.
Non-uniform switching sequences also appear in the field of compressed
sampling~\cite{Bryan2013,Candes2008,Nagahara2012}, where there is a given number $m$ of samples
which, in general, are not uniformly spaced in time either. The objective is to recover the full signal in 
the receptor side. This technique has shown to be promising in band-limited networked control
applications~\cite{Nagahara2012b}.

%
%

Motivated by the applications described above, we restrict the control actions
to the class of piecewise constant functions of time. We also assume that the
sequence of switching times is known \emph{a priori}, inasmuch as the choice
of the switching times can be carried out independently from the choice of the
control levels when performing bandwidth optimization~\cite{Skafidas1998,Xu2004}. 



The purpose of this paper is to derive an optimal min-max control strategy in an \emph{analytical} fashion,
for the case when parametric uncertainty belongs to a finite set. That is,
when there is a finite set of possible models (see Remark \ref{remark:convexSet} ). The LQ optimal control
problem is stated formally in Section~\ref{sec:ProblemFormulation}.
Motivated by the piecewise constant nature of the control laws, the problem
will be reformulated in a discrete-time context.

\paragraph*{Main Contribution.}
Theorem~\ref{th:Contribution} in Section~\ref{sec:main} states the
solution to the LQ problem. Using convex analysis and generalized gradients, a
discrete-time extended Riccati equation is obtained. The equation is parametrized
by $\mu$, a vector whose elements are convex multipliers for the costs incurred by the
individual plants. To compute the optimal control, a maximization problem constrained
to a finite dimensional simplex has to be solved for $\mu$. Additionally, it is shown
that a complementary slackness condition holds between $\mu$ and the individual costs.

In Section~\ref{sec:numAlg} we propose a numerical algorithm to find $\mu$ and the
corresponding optimal control, this algorithm can be used in
conjunction with MPC . It is based on the classical gradient with projection
approach, but the gradient is approximated using Kiefer-Wolfowitz algorithm~\cite{Poznyak2009}.
The complementary slackness condition turns out to be critical in deciding a practical
stopping criterion for the numerical algorithm.

Finally, we present numerical examples illustrating the usefulness of Theorem~\ref{th:Contribution},
the feasibility of the numerical algorithm and the soundness of the min-max paradigm when confronted to multi-models.



\section{Problem Formulation} \label{sec:ProblemFormulation}

Consider the following general continuous-time linear system with parametric uncertainty:
\begin{equation} \label{eq:ContinuousSys}
 \dot{x}(t) = A^{\alpha}(t) x(t) + B^{\alpha}(t) u(t) \;, \quad x(t_0) = x_0 \;, \quad \alpha \in \ACal \;, \quad t \in [t_0, t_N] \;,
\end{equation}
where $x(t) \in \mathbb{R}^{n}$ is the system state vector and $u(t) \in \mathbb{R}^{m}$ 
is the control input.
The actual realization of the system and input matrices $A^\alpha(t)$ 
and $B^\alpha(t)$ is unknown, but these belong to a known finite set 
which is indexed by $\alpha$. In other words, equation~\eqref{eq:ContinuousSys} 
describes one realization taken out from a finite set of possible systems. 
The index $\alpha$ is taken to be constant, \emph{but unknown},
through all the process lifetime, i.e., from $t_0$ until $t_N$.
For a fixed $\alpha \in \ACal$, we will denote the solution of~\eqref{eq:ContinuousSys} 
by $x^{\alpha}(t)$.

\subsection{Admissible Controls}

The family of admissible controls takes a stepwise form. 
This is motivated by some applications in networked control systems.
For example, when passing through a digital band-limited communication 
channel, the control input $u$ may only take a finite number of changes 
(switches) on its levels over the whole interval $[t_0, t_N]$. 
Let the times at which these switches occur be given by a monotonically 
increasing sequence\footnote{Note that, since the controller is updated 
at most $N-1$ times, a reduction on $N$ directly translates into a reduction 
of the required band-width of the communication channel.}
\begin{equation} \label{eq:delta}
 \delta = \{t_0, t_1, \ldots , t_{N-1} \} \;,
\end{equation}
where $\delta$ is bounded by $t_N$. The set of admissible controls is given by
\begin{multline} \label{eq:UAdms}
 \mathcal{U}_{\mathrm{ad}}^\delta = \Big\lbrace u:[t_0, t_N] \rightarrow \mathbb{R}^{m} \:\big | \: 
  u(t) = \textstyle\sum\limits_{k = 0}^{N-1} \chi_{[t_k, t_{k+1})}(t) v_k, \\ t_k \in \delta, \; 
   v_k \in \mathbb{R}^{m} , \; k = 0,\ldots,N-1 \Big\rbrace 
\end{multline}
with
\begin{displaymath}
 \chi_{[t_k,t_{k+1})}(t) = 
  \begin{cases} 1, & \text{if } t \in [t_k, t_{k+1}) \\
                0, & \text{ otherwise}
  \end{cases}
\end{displaymath}
the characteristic function of the interval $[t_k, t_{k+1})$ and $v_k$ the value 
of $u(t)$ when $t \in [t_k, t_{k+1})$. It is worth mentioning that the switching
sequences $\delta$ that are taken under consideration are non uniform in general.
The reason is that, regarding minimal bandwidth consumption, a uniform sampling
is not always the best choice (see e.g.~\cite{Bini2014} and Remark \ref{remark:delta}
below).

\subsection{Problem Statement}

Associated to~\eqref{eq:ContinuousSys}, we consider the quadratic cost functional
\begin{equation} \label{eq:ContinuousCost}
 J^{\alpha}(u) = \dfrac{1}{2} x^{\alpha}(t_N)^{\top} G x^{\alpha}(t_N) + 
  \dfrac{1}{2} \int_{t = 0}^{t_N} x^{\alpha}(t)^{\top} Q(t) x^{\alpha}(t) + 
  u(t)^{\top} R(t) u(t) \mathrm{d}t \;,
\end{equation}
where the usual positiveness conditions are assumed:
\begin{equation} \label{eq:positive}
 G = G^{\top} \geq 0 \;, \quad Q(t) = Q(t)^{\top} \geq 0 \quad \text{and} \quad R(t) = R(t)^{\top} \geq \varepsilon I 
\end{equation}
for all $t \in [t_0, t_N]$ and some $\varepsilon > 0$.

The optimization problem consists in finding a control action 
$u^{*}(\cdot) \in \mathcal{U}_{\mathrm{ad}}^\delta$ that provides a ``good''
behavior for all systems from the given collection of models~\eqref{eq:ContinuousSys},
including the ``worst'' case. The resulting control strategy is applied to~\eqref{eq:ContinuousSys},
regardless of the actual $\alpha$-realization. We state this formally as the min-max
optimization problem
\begin{equation} \label{eq:minMaxProblem}
\begin{split}
& \text{minimize } J(u) \\
& \text{subject to } u(\cdot) \in \mathcal{U}_{\mathrm{ad}}^\delta \;,
\end{split}
\end{equation}
where $J(u) = \max_{\alpha \in \ACal} J^{\alpha}(u)$ subject to \eqref{eq:ContinuousSys} 
(see, e.g.,~\cite{Azhmyakov2010} and~\cite{Boltyanski2011}). 

\begin{remark} \label{remark:delta}
All admissible controls $u(\cdot)$ clearly depend on the switching sequence $\delta$ but, as it
is noted in~\cite{Skafidas1998} and~\cite{Xu2004}, it is possible to perform optimization in two stages:
First, seek an optimal switching sequence $\delta$ and then look for an optimal control input for the
given switching sequence. We focus on the second stage, so it is assumed that the switching sequence
is given \emph{a priori}.
\end{remark}
\begin{remark} \label{remark:convexSet}
It is worth to mention that, for the case when the set of considered 
uncertainties is a compact polyhedron (i.e, we have a infinite number
of possible realizations of \eqref{eq:ContinuousSys}), and the map 
$\alpha \mapsto J^{\alpha}(u)$ is convex, the problem is reduced to
study only those uncertainties what belongs to the set of extreme
points of the polyhedron (i.e., its vertices) \cite[Theorem 3.4.7]
{Bazaraa2006}. In this case the problem falls in our setting 
(see \cite{Ramirez2002} for an application of this fact).
\end{remark}

The stepwise nature of the control motivates the use of a discrete-time approach, using $\delta$ as the sampling
time-sequence. Indeed, the problem is finite-dimensional, since it is only necessary to find $N$ values of the control signal.
Now, let 
\begin{equation} \label{eq:decision}
 v = \begin{bmatrix} v_0^{\top} & \cdots & v_{N-1}^{\top} \end{bmatrix}^{\top} \in \mathbb{R}^{m N}
\end{equation}
be the decision vector and allow us construct the vector $x^\alpha$, obtained by appending each point of the resulting
discrete-time orbit,
\begin{displaymath}
 x^\alpha = \begin{bmatrix} x_1^{\alpha {\top}} & x_2^{\alpha {\top}} & \cdots & x_N^{\alpha {\top}} \end{bmatrix}^{\top} \in \mathbb{R}^{nN} \;.
\end{displaymath} 
The discrete-time representation of problem~\eqref{eq:minMaxProblem} is thus
\begin{equation} \label{eq:minMaxDisc}
 \text{minimize } \bar{J}(v, x, x_0) \;,
\end{equation}
where the discrete costs are given by
\begin{displaymath}
 \bar{J}(v,x,x_0) = \max_{\alpha \in \ACal} \bar{J}^{\alpha}(v,x^\alpha,x_0)
\end{displaymath}
subject to
\begin{displaymath}
 x_{k+1}^{\alpha} = \Phi^{\alpha}_{k} x^{\alpha}_k + \Gamma^{\alpha}_{k}v_k \;, \quad x^{\alpha}_0 = x_0 \;, \quad \alpha \in \ACal \;, \quad k= 0,\ldots,N-1
\end{displaymath}
with
\begin{equation}
 \bar{J}^{\alpha}(v,x^\alpha,x_0) = \dfrac{1}{2} x^{\alpha {\top}}_{N} G x^{\alpha}_N + \dfrac{1}{2} 
  \sum_{k = 0}^{N-1} \left[ x^{\alpha {\top}}_k \Pi_k^{\alpha} x^{\alpha}_k + 2 x^{\alpha {\top}}_k 
  \Theta^{\alpha {\top}}_k v_k + v_k^{\top} \Psi_k^{\alpha} v_k \right] \label{eq:DiscCrossFunc} \;.
\end{equation}
%
To alleviate the notation, we will denote the cost functional 
$\bar{J}^{\alpha}(v, x^{\alpha}, x_0)$ simply by $\bar{J}^{\alpha}(v, x^{\alpha})$, bearing 
in mind the dependence of the initial condition in the cost\footnote{Moreover, we will write
in~\eqref{eq:W} the cost as a function of $v$ only, to reflect the fact that $x^\alpha$ depends
on $v$ through the constraint equations.}.

It is straightforward to verify that the weighting matrices in~\eqref{eq:DiscCrossFunc} are
\begin{align*}
 \Gamma^{\alpha}(t,t_k) &:= \int_{t_k}^{t} \Phi^{\alpha}(t,\tau) B^{\alpha} (\tau) \mathrm{d} \tau \in \mathbb{R}^{n \times m} \;, \\
         \Pi_k^{\alpha} &:= \int_{t_k}^{t_{k+1}} \Phi^{\alpha}(t,t_k)^{\top} Q(t)  \Phi^{\alpha}(t,t_k) \mathrm{d}t \in \mathbb{R}^{n \times n} \;, \\
      \Theta_k^{\alpha} &:= \int_{t_k}^{t_{k+1}} \Gamma^{\alpha}(t,t_k)^{\top} Q(t) \Phi^{\alpha}(t,t_k) \mathrm{d}t \in \mathbb{R}^{m \times n} \;, \\
        \Psi_k^{\alpha} &:= \int_{t_k}^{t_{k+1}} \left( \Gamma^{\alpha}(t,t_k)^{\top} Q(t) \Gamma^{\alpha}(t,t_k) + R(t) \right) \mathrm{d}t \in \mathbb{R}^{m \times m} \;, \\
    \Gamma^{\alpha}_{k} &:= \Gamma^{\alpha}(t_{k+1}, t_k)
\end{align*}
and $\Phi^{\alpha}_{k} := \Phi^{\alpha}(t_{k+1},t_k)$ with $\Phi^{\alpha}(t,t_0)$ the state transition matrix for
the $\alpha$-system~\eqref{eq:ContinuousSys}. Notice that $\Pi_k^{\alpha} = \Pi_k^{\alpha {\top}} \geq 0$ and
$\Psi_k^{\alpha} = \Psi_k^{\alpha {\top}} > 0$ for all $k = 0,\ldots, N-1$.


\section{Solution to the Min-Max Problem} \label{sec:Min-maxSolution}

\subsection{Multi-model method of Lagrange multipliers}

Recall that, in the classical discrete-time LQ problem (i.e., when $\ACal$ 
is a singleton, $\ACal = \{ 1 \}$), the optimal control can be 
obtained using the Lagrange multipliers framework~\cite[ch.~8]{Ogata1995}.
The following lemma states that the framework is still valid, \emph{mutatis mutandis},
for the discrete-time multi-model case. 

\begin{lemma} \label{lemma:MultiLagrange}
Let $(v^{*}, x^{*})$ be an optimal pair solution of~\eqref{eq:minMaxDisc}, then the following
conditions are satisfied:
\begin{subequations} \label{eq:Opt}
\begin{align}
 0 & \in \Convex \{ \nabla_{v} L^{\alpha} (v^* ,x^{*}, \lambda^{\alpha}) : \alpha \in I(v^*)\} \label{eq:Opt_1} \\
 0 & = \nabla_{x} L^{\alpha}(v^*, x^{*}, \lambda^{\alpha}) \text{ for all } \alpha \in I(v^*) \label{eq:Opt_2} \\
 0 & = \nabla_{\lambda} L^{\alpha}(v^*, x^{*}, \lambda^{\alpha}) \text{ for all } \alpha \in I(v^*) \label{eq:Opt_3} \;,
\end{align}
\end{subequations}
where $\Convex$ denotes convex closure and $L^{\alpha}: \RE^{mN} \times \RE^{nN} \times \RE^{nN} \rightarrow \RE$
denotes the Lagrangian of the problem, that is,
\begin{equation} \label{eq:Lagrangian}
 L^{\alpha} (v, x^{\alpha }, \lambda^{\alpha}) = \bar{J}^{\alpha}(v, x^{\alpha}) + 
  \lambda^{\alpha \top} g^{\alpha}(v, x^{\alpha})
\end{equation}
with
\begin{align*}
 g^{\alpha}(v, x^{\alpha}) &= \left[g^{\alpha \top}_0(v, x^{\alpha}), \ldots, 
  g^{\alpha \top}_{N-1}(v, x^{\alpha}) \right]^{\top} \in \RE^{nN} \\ 
 g^{\alpha}_k(v, x^{\alpha}) &= -x_{k+1}^{\alpha} + 
  \Phi^{\alpha}_{k} x^{\alpha}_k + \Gamma^{\alpha}_{k} v_k \;.
\end{align*} 
$I(v)$ denotes the set of indices where $\bar{J}^{\alpha}$ 
reaches the maximum value (with $v$ fixed), i.e.,
\begin{displaymath}
 I(v) = \{ \alpha \in  \ACal : \bar{J}^{\alpha}(v, x^{\alpha}) = \bar{J}(v, x) \} \;.
\end{displaymath}
\end{lemma}

\begin{proof}
In order to make the proof easier to follow, we introduce the function
\begin{equation} \label{eq:W}
 W^{\alpha}(v) := \bar{J}^{\alpha}(v, x^{\alpha}(v)) \;,
\end{equation}
where we are exploiting the fact that $x^{\alpha}$ is a function of $v$ through 
the restriction equations. Indeed, is easy to see that for each $v$, $x^{\alpha}$
is uniquely defined by the constraint $g^{\alpha}(v, x^{\alpha}) = 0$. Moreover,
$W^{\alpha}(\cdot)$ is again convex since it is the composition of an affine and
a convex function.
  
Our immediate goal is to derive necessary conditions for optimality of the problem
\begin{displaymath}
 \text{minimize: } \bar{J}(v, x) \;,
\end{displaymath}
which is equivalent to
\begin{equation} \label{eq:Cost_v}
\text{minimize } W(v)
\end{equation}
with $W(v) = \max_{\alpha \in \ACal} W^{\alpha}(v)$. Note that $W$ is 
obtained by choosing the maximum among a finite set of functions, so it is in general
a non-smooth function, even in the case where each $W^{\alpha}(v)$ is smooth. 

Is well know~\cite[p.~70]{Makela1992} that a necessary condition for optimality 
of~\eqref{eq:Cost_v} is $0 \in \partial_c W(v^{*})$, where $\partial_c W$ 
denotes the convex subdifferential of $W$ at $v^{*}$\footnote{Note that $W$ is
the maximum of a finite set of convex functions, which is again convex.}.
The subdifferential of the max function has been reported in the literature 
(see e.g.~\cite[p. 49]{Makela1992}) and it is known to satisfy $\partial_c W(v) 
= \Convex \left\lbrace \nabla W^{\alpha}(v) : \alpha \in I(v) \right\rbrace$,
where $I(v) = \{ \alpha \in \ACal : W^{\alpha}(v) = W(v) \}$ 
(i.e., $I(v)$ is the set of indices where the maximum is reached). Thus, we have 
the following necessary condition for optimality
\begin{displaymath}
 0 \in \Convex \left\lbrace \nabla W^{\alpha}(v^{*}) : \alpha \in I(v^{*}) \right\rbrace \;.
\end{displaymath}
It is worth stressing the fact that the right-hand side does not involve all the possible
$\alpha$-realizations of the uncertain plant. Only those plants for which the maximum is
reached play a role in the optimization procedure --- In the following section we propose 
a method  for finding these \emph{extreme} plants (see also Remark~\ref{remark:3}).

From the definition of $W^{\alpha}$ we have
\begin{displaymath}
 \nabla W^{\alpha}(v) = \nabla_v \bar{J}^{\alpha}(v,x^{\alpha}) + 
  \nabla_x \bar{J}(v,x^{\alpha}(v)) \nabla{x^{\alpha}(v)} \;,
\end{displaymath}
where we defined the gradient of a map $F: \RE^{n} \rightarrow \RE^{m}$ as
\begin{displaymath}
 \nabla F(x) = \begin{bmatrix}
  \frac{\partial F_1(x)}{\partial x_1} & \cdots & \frac{\partial F_1(x)}{\partial x_n} \\
  \vdots & \ddots & \vdots \\
  \frac{\partial F_m(x)}{\partial x_1} & \cdots & \frac{\partial F_m(x)}{\partial x_n}
\end{bmatrix} \in \RE^{m \times n}
\end{displaymath}
(i.e., we adhere to the convention that the gradient of a scalar function is as a row vector).

Since the plant is a well defined system, there exist at least one pair $(\bar{v}, \bar{x}^{\alpha})$
for which $g^{\alpha}(\bar{v}, \bar{x}^{\alpha}) = 0$. Furthermore,
\begin{displaymath}
 \nabla_x g^{\alpha} (v,x^{\alpha}) =
  \begin{bmatrix} 
              -I_n &               0 &      0 & \cdots &                   0 & 0 \\
   \Phi_1^{\alpha} &            -I_n &      0 & \cdots &                   0 & 0 \\
                 0 & \Phi_2^{\alpha} &   -I_n & \cdots &                   0 & 0 \\
            \vdots &          \vdots & \vdots & \ddots &              \vdots & \vdots \\
                 0 &               0 &      0 & \cdots &                -I_n & 0 \\
                 0 &               0 &      0 & \cdots & \Phi_{N-1}^{\alpha} & -I_n
 \end{bmatrix}
\end{displaymath}
is an invertible matrix, so it is possible to apply the implicit function theorem,
write $x^{\alpha}$ as a function of $v$ and write the gradient
\begin{displaymath}
 \nabla x^{\alpha}(v) = - \left[ \nabla_x g^{\alpha}(v,x^{\alpha}(v)) \right]^{-1} 
 \nabla_v g^{\alpha}(v,x^{\alpha}(v)) \;.
\end{displaymath}
This results in
\begin{multline*}
 \nabla W^{\alpha} (v) = \nabla_v \bar{J}^{\alpha} (v,x^{\alpha}(v)) 
  - \\
  \nabla_x \bar{J}^{\alpha} (v,x^{\alpha}(v)) 
   \left[ \nabla_x g^{\alpha}(v,x^{\alpha}(v)) \right]^{-1} 
\nabla_v g^{\alpha}(v,x^{\alpha}(v)) \;.
\end{multline*}
Now, define $\lambda^{\alpha \top} := -\nabla_x \bar{J}^{\alpha}
(v,x^{\alpha}(v)) \left[ \nabla_x g^{\alpha}(v,x^{\alpha}(v)) 
\right]^{-1} \in \RE^{1 \times nN}$, write the Lagrangian 
as in \eqref{eq:Lagrangian} and immediately obtain~\eqref{eq:Opt_1}. From the
definition of $\lambda^{\alpha}$, we have
\begin{displaymath}
 0 = \nabla_x \bar{J}^{\alpha}(v, x^{\alpha}(v)) + 
  \lambda^{\alpha \top} \nabla_v g^{\alpha}(v, x^{\alpha}(v)) \;,
\end{displaymath}
which is the same as~\eqref{eq:Opt_2}. Finally, note that~\eqref{eq:Opt_3} is a restatement
of the constraints.
\end{proof}

\subsection{Extended Riccati equation, complementary slackness} \label{sec:main}

Throughout the rest of this section we will work simultaneously with all $\alpha$-realizations. In order to
make the presentation clearer, we introduce the following extended matrices:
\begin{align*}
        \mathbf{G}(\mu) &= \Diag \{ \mu_1 G, \ldots , \mu_{|\ACal|} G \}
         \in \mathbb{R}^{n |\ACal| \times n |\ACal|} \\
    \mathbf{\Pi}_k(\mu) &= \Diag \left\lbrace \mu_1 \Pi_k^{1} , \ldots , 
     \mu_{|\ACal|} \Pi_k^{|\ACal|} \right\rbrace \in \mathbb{R}^{n |\ACal| \times n |\ACal|} \\
 \mathbf{\Theta}_k(\mu) &= 
  \begin{bmatrix} 
   \mu_1 \Theta^{1}_k & \cdots & \mu_{|\ACal|} \Theta^{|\ACal|}_k 
  \end{bmatrix} \in \mathbb{R}^{m \times n |\ACal|} \\
   \mathbf{\Psi}_k(\mu) &= \sum_{\alpha = 1}^{|\ACal|} \mu_{\alpha} \Psi^{\alpha}_k \in \mathbb{R}^{m \times m} \\
        \mathbf{\Phi}_k &= \Diag \{ \Phi^{1}_k, \ldots, \Phi^{|\ACal|}_k  \} \in \mathbb{R}^{n |\ACal| \times n |\ACal|} \\
      \mathbf{\Gamma}_k &= 
        \begin{bmatrix} 
         \Gamma^{1 \top}_k & \Gamma^{2 \top}_k & \cdots & \Gamma^{|\ACal| \top}_k
        \end{bmatrix}^{\top} 
        \in \mathbb{R}^{n |\ACal| \times m} \\
      \mathbf{M}(\mu^*) &= \Diag \left( \mu_1^* I_n,\ldots,\mu_{|\ACal|}^*I_n \right) \;,
\end{align*}
where $|\ACal|$ is the cardinality of $\ACal$. Note that the vector $\mu$ (formally defined
below as an element of a simplex) only intervenes in the matrices that are related to the cost.
Also, note that the symmetry and positive (semi) definiteness
of $G$, $\Pi_\alpha^k$ and $\Psi_\alpha^k$ are inherited by 
$\mathbf{G}(\mu)$, $\mathbf{\Pi}_k(\mu)$ and $\mathbf{\Psi}_k(\mu)$,
for all $k \in 0,\ldots,N-1$ and for all $\mu$ in the simplex.

Let us now define the extended vectors
\begin{displaymath}
 \mathbf{x}_k = \begin{bmatrix} x^{1}_k \\ \vdots \\ x^{|\ACal|}_k \end{bmatrix} \in \mathbb{R}^{n |\ACal|} \quad \text{and} \quad 
  \mathbf{\Lambda}_k(\mu^{*}) = \begin{bmatrix} \mu^{*}_1 \lambda_k^{1} \\ \vdots \\ \mu^{*}_{|\ACal|} \lambda_k^{|\ACal|} \end{bmatrix} \in \mathbb{R}^{n |\ACal|} \;,
\end{displaymath}
so that we can formulate the main contribution of this work. The result can be interpreted as a 
discrete-time version of the robust maximum principle developed by one of the authors in~\cite{Boltyanski2011}
and applied to the LQ problem.

\begin{theorem}\label{th:Contribution}
Consider the multi-model linear system~\eqref{eq:ContinuousSys} and the cost functional~\eqref{eq:ContinuousCost} subject to
the usual positiveness assumptions~\eqref{eq:positive}. Let $\delta$ be a switching sequence given a priori and of the form~\eqref{eq:delta}
and let $\mathcal{U}_{\mathrm{ad}}^\delta$ be the set of admissible controllers defined by~\eqref{eq:UAdms}. The control
$u^{**}(\cdot) \in \mathcal{U}_{\mathrm{ad}}^\delta$ that solves the min-max problem~\eqref{eq:minMaxProblem} is given by:
\begin{equation} \label{eq:th:OptimalControl}
\begin{split}
 u^{**}(t) & := u^{*}(t,\mu^*) = \sum_{k = 0}^{N-1} \chi_{[t_k, t_{k+1})} (t) v_k^*(\mu^*) \;, \quad t_k \in \delta \;, \\
  v_k^*(\mu) &= - \left( \mathbf{\Psi}_k(\mu) + \mathbf{\Gamma}_k^{\top} \mathbf{P}_{k+1}(\mu) 
   \mathbf{\Gamma}_k \right)^{-1} \left( \mathbf{\Theta}_k(\mu) + \mathbf{\Gamma}_k^{\top}
   \mathbf{P}_{k+1}(\mu) \mathbf{\Phi}_k \right) \mathbf{x}_k,
\end{split}
\end{equation}
where the boldface matrices were defined previously, except for $\mathbf{P}_{k}(\mu)$, which is defined implicitly 
as the positive-definite solution of the discrete-time Riccati equation
\begin{subequations} \label{eq:th:Riccati}
\begin{multline}
 \mathbf{P}_k(\mu) = \mathbf{\Pi}_k(\mu) + \mathbf{\Phi}_k^{\top} \mathbf{P}_{k+1}(\mu) \mathbf{\Phi}_k - 
  \left(  \mathbf{\Theta}_k(\mu) + \mathbf{\Gamma}_k^{\top} \mathbf{P}_{k+1}(\mu) \mathbf{\Phi}_k \right)^{\top} \times \\
  \left( \mathbf{\Psi}_k(\mu) + \mathbf{\Gamma}_k^{\top} \mathbf{P}_{k+1}(\mu) \mathbf{\Gamma}_k \right)^{-1} 
  \left(  \mathbf{\Theta}_k(\mu) + \mathbf{\Gamma}_k^{\top} \mathbf{P}_{k+1}(\mu) \mathbf{\Phi}_k \right)
\end{multline}
with boundary condition
\begin{equation} 
 \mathbf{P}_N(\mu) = \mathbf{G}(\mu) \;.
\end{equation}
\end{subequations}
The optimal vector $\mu^{*}$ is the solution of
\begin{equation} \label{eq:th:MaxProblem}
 \max_{\mu \in \SCal^{|\ACal|}} \dfrac{1}{2} \mathbf{x}_0^{\top} \mathbf{P}_0(\mu) \mathbf{x}_0
\end{equation} 
and $\SCal^{|\ACal|}$ is the simplex
of dimension $|\ACal|-1$,
\begin{displaymath}
 \SCal^{|\ACal|} = \left\lbrace \mu \in \mathbb{R}^{|\ACal|} : 
  \sum_{\alpha \in \ACal} \mu_{\alpha} = 1, \; \mu_{\alpha} \geq 0, \; \text{for all } \alpha \in \ACal \right\rbrace \;.
\end{displaymath}

Moreover, the complementary slackness condition
\begin{equation}\label{eq:th:SlacknessCondition}
 \mu^{*}_{\alpha} \left[ J^{\alpha}(u^{**}) - J(u^{**}) \right] = 0
\end{equation}
is satisfied for every $\alpha \in \ACal$.  
\end{theorem}
\begin{remark} \label{rem:convCost}
Equations~\eqref{eq:th:OptimalControl} and~\eqref{eq:th:Riccati} define the solution of a classical
discrete-time LQ problem for the extended system $\mathbf{x}_{k+1} = \mathbf{\Phi}_k \mathbf{x}_k + \mathbf{\Gamma}_k v_k$,
the only difference being the dependence that the matrices have on $\mu$. In view of this observation
and the structure of the extended cost matrices, one concludes that the cost for the extended system is a convex combination
of the individual costs. In symbols,
\begin{multline*}
 \dfrac{1}{2} \mathbf{x}_N^{\top} \mathbf{G}(\mu) \mathbf{x}_N + 
  \dfrac{1}{2} \sum_{k = 0}^{N-1} \left( \mathbf{x}_k^{\top} 
  \mathbf{\Pi}_k(\mu) \mathbf{x}_k +
  2 \mathbf{x}_k \mathbf{\Theta}_k(\mu)^{\top} v_k + 
  v_k^{\top} \mathbf{\Psi}_k(\mu) v_k \right) \\ = \sum_{\alpha \in \ACal} \mu_{\alpha} \bar{J}^{\alpha}(v, x^{\alpha}) \;.
\end{multline*}
\end{remark}
\begin{remark} \label{remark:3}
The complementarity slackness condition reveals that only the \emph{extreme} plants
(those for which the maximum is attained) play a role in the computation of the optimal
control. Indeed, the optimal vector $\mu^{*}$ acts as an indicator for the set of 
extreme plants (having nonzero elements in the $\alpha$-coordinates if, and only if, the
corresponding plant is extreme).
\end{remark}
%
%
\begin{proof}[Proof of Theorem~\ref{th:Contribution}]
It has already been established that the continuous-time 
problem~\eqref{eq:minMaxProblem} is equivalent to the
discrete-time problem \eqref{eq:minMaxDisc}. 
Thus, our problem consists in minimizing $\bar{J}(v, x)$.
Lemma \ref{lemma:MultiLagrange} states that the optimal control pair must satisfy~\eqref{eq:Opt}, which translates to
\begin{subequations} \label{eq:optimConditions1}
\begin{equation}
 0 = \sum_{\alpha \in I(v^{*})} \left\lbrace \mu_{\alpha}^{*} 
  \left( \nabla_{v} \bar{J}^{\alpha}(v^{*}, x^{*}) + \lambda^{\alpha \top} \nabla_v g^{\alpha}(v^{*}, x^{*}) \right) \right\rbrace 
\end{equation}
for some $\mu^{*} \in \SCal^{|I(v^{*})|}$ (i.e, such that $\sum_{\alpha \in I(v^{*})} \mu_{\alpha}^{*} = 1$ and $\mu^{*}_{\alpha} \geq 0$) and
\begin{align}
 0 &= \nabla_{x} \bar{J}^{\alpha}(v^{*}, x^{*}) + 
 \lambda^{\alpha \top} \nabla_x g^{\alpha}(v^{*}, x^{*}) \\
 0 &= g^{\alpha}(v^{*}, x^{*})
\end{align}
\end{subequations}
for all $\alpha \in I(v^*)$. Allow us to embed $\mu^*$ in the larger simplex $\SCal^{|\ACal|}$ by putting zeros in the
$\alpha \notin I(v^{*})$-entries, so the new vector $\mu^{*} \in \SCal^{|\ACal|}$ acts an indicator for the \textit{extreme} plants.

Computing~\eqref{eq:optimConditions1} explicitly gives
\begin{equation} \label{eq:optimConditions2}
\begin{split}
 0 &= \sum_{\alpha \in \ACal} \mu^{*}_{\alpha} 
 \left( \Psi_k^{\alpha} v_k^{*}+ \Theta_k^{\alpha} x_k^{\alpha *} + 
 \Gamma_k^{\alpha \top} \lambda_{k+1}^{\alpha} \right) \\
 0 &= \mu_{\alpha}^{*} \left[ \Pi_k^{\alpha} x_k^{\alpha *} + 
 \Theta_k^{\alpha \top} v_k^{*} + \Phi_{k}^{\top} \lambda_{k+1}^{\alpha} - 
 \lambda_k^{\alpha} \right] \\
 0 &= \mu_{\alpha}^{*} \left[\lambda_N - G x_N^{\alpha *} \right] \\
 0 &= \mu_{\alpha}^{*} \left[ -x_{k+1}^{\alpha *} + 
 \Phi_k^{\alpha} x_k^{\alpha *} + \Gamma_k^{\alpha} v_k^{*} \right]
\end{split}
\end{equation}
for all $k = 0,\ldots,N-1$. In order to make the notation more compact we will make
use of the block matrices defined above. The optimality conditions~\eqref{eq:optimConditions2}
can now be rewritten as
\begin{subequations} \label{eq:extLagr}
\begin{align}
 0 &= \mathbf{\Psi}_{k}(\mu^{*}) v_k^{*} + \mathbf{\Theta}_k(\mu^{*}) \mathbf{x}_k^{*}  + 
  \mathbf{\Gamma}_k^{\top} \mathbf{\Lambda}_{k+1}(\mu^{*})  \\ 
 0 &= \mathbf{\Pi}_k(\mu^{*}) \mathbf{x}_k^{*} + \mathbf{\Theta}_k(\mu^{*})^{\top} v_k^{*} + 
  \mathbf{\Phi}_k^{\top} \mathbf{\Lambda}_{k+1}(\mu^{*}) - \mathbf{\Lambda}_k(\mu^{*}) \\
 0 &= \mathbf{G}(\mu^{*}) \mathbf{x}_N^{*} - \mathbf{\Lambda}_N(\mu^{*}) \\
 0 &= \mathbf{M}(\mu^*)[-\mathbf{x}_{k+1}^{*} + \mathbf{\Phi}_k \mathbf{x}_k^{*} + \mathbf{\Gamma}_k v_k^{*}]
\end{align}
\end{subequations}
for all $k = 0, \ldots , N-1$. Note that these equations are similar to the classical discrete-time LQ
optimality equations of an $n |\ACal|$-dimensional linear system parametrized by $\mu^{*}$ \cite[p.~582]{Ogata1995}.
To obtain the optimal control $v^{*}$, we simply follow the classical discrete-time approach
(only the main steps are reported since the approach is well known).

Recall that it is always possible to write the adjoint variable $\mathbf{\Lambda}(\mu^*)$ as a linear
function of the (extended) state
$\mathbf{x}^{*}$, i.e., as $\mathbf{\Lambda}_k(\mu^{*}) = \mathbf{P}_k(\mu^{*}) \mathbf{x}_k^{*}$.
It is straightforward to verify that, for $\mathbf{\Lambda}_k(\mu^{*})$ to satisfy~\eqref{eq:extLagr},
$\mathbf{P}_k(\mu^{*})$ must be the (positive semi-definite) solution of the
$\mu^{*}$-parametrized Riccati equation~\eqref{eq:th:Riccati}. The optimal control then takes the form
\begin{displaymath}
 v_k^{*}(\mu^{*}) = - \left( \mathbf{\Psi}_k(\mu^{*}) + 
 \mathbf{\Gamma}_k^{\top} \mathbf{P}_{k+1}(\mu^{*}) \mathbf{\Gamma}_k \right)^{-1} 
  \left( \mathbf{\Theta}_k(\mu^{*}) + \mathbf{\Gamma}_k^{\top} 
  \mathbf{P}_{k+1}(\mu^{*}) \mathbf{\Phi}_k \right) \mathbf{x}_k^{*}
\end{displaymath}
(see~\cite[ch.~8]{Ogata1995} for a detailed development of the classical discrete-time LQ problem).

It is not difficult to prove that the optimal cost of the extended system is equal to
$\frac{1}{2} \mathbf{x}_0^{\top} \mathbf{P}_k(\mu^{*}) \mathbf{x}_0$ (see~\cite[p. 575]{Ogata1995}) so,
according to Remark~\ref{rem:convCost},
\begin{equation} \label{eq:optCost}
 \sum_{\alpha \in \ACal} \mu_{\alpha}^* \bar{J}^{\alpha}(v^{*}(\mu^*), x^{\alpha *}) = 
  \frac{1}{2} \mathbf{x}_0^{\top} \mathbf{P}_k(\mu^{*}) \mathbf{x}_0 \;.
\end{equation}
On the other hand, according to Lemma~\ref{lemma:EqProblems} in the Appendix, this vector also satisfies
\begin{equation} \label{eq:muDef}
 \mu^{*} = \Argmax_{\mu \in \SCal^{|\ACal|}} \sum_{\alpha \in \ACal} 
  \mu_{\alpha} \bar{J}^{\alpha}(v^{*}(\mu), x^{\alpha *}) \;.
\end{equation}
Statement~\eqref{eq:th:MaxProblem} now follows directly from~\eqref{eq:optCost} and~\eqref{eq:muDef}.

Lemma~\ref{lemma:EqProblems} also implies that the components of $\mu^{*}$ will be different 
from zero at the $\alpha$-positions where $J^{\alpha}(u^{**}) = J(u^{**})$
(i.e., where the maximum is attained) and will be zero otherwise. This fact is equivalent to the
complementary slackness condition~\eqref{eq:th:SlacknessCondition}.
\end{proof}

\begin{remark} \label{remark:4}
Notice from~\eqref{eq:th:MaxProblem} that $\mu^{*}$ depends on the initial conditions and
the system parameters only (not on the whole state trajectory).
This allow us to separate the optimization problem in two simpler subproblems. Namely, the first
part consists in solving the $\mu$-parametrized Riccati equation~\eqref{eq:th:Riccati}. 
The second part consists in finding the solution $\mu^{*}$ of~\eqref{eq:th:MaxProblem}. 
Both stages can be accomplished off-line. 
\end{remark}

\subsection{Numerical Algorithm} \label{sec:numAlg}

Theorem~\ref{th:Contribution} provides the feedback control equations
in terms of the parameters of every $\alpha$-model and $\mu^{*}$. 
Thus, in order to determine the control law completely, it is necessary
to solve~\eqref{eq:th:MaxProblem}.
Gradient-based algorithms are widely used for numerical optimization. These are methods
where the search directions are defined by the gradient of a target function at the
current iteration point. Computing the gradient of the performance index
$\frac{1}{2} \mathbf{x}_0^{\top} \mathbf{P}_0(\mu) \mathbf{x}_0$ directly 
is a challenging task, mainly because of the recursion in the Riccati 
equation~\eqref{eq:th:Riccati}. To circumvent this problem we propose 
to use an approximation of the gradient in combination with a projection 
algorithm that guarantees that $\mu$ belongs to $\SCal^{|\ACal|}$.

Let $f(\mu) \in \mathbb{R}$ be a convex cost function to be minimized. The classical gradient with projection algorithm for finding the
minimum is given by the recursion~\cite{Polyak1965}
\begin{equation} \label{eq:muClassical}
 \mu^{j} = \operatorname{Proj} \left[ \mu^{j-1} - \gamma_j \nabla f(\mu^{j-1}) \right]_{\SCal^{|\ACal|}} \;,
\end{equation}
where $\operatorname{Proj}\left[ \cdot \right]_{\SCal^{|\ACal|}}$ refers to the 
projection of a point in $\mathbb{R}^{|\ACal|}$ into the set $\SCal^{|\ACal|}$, i.e.,
\begin{displaymath}
 x^{*} = \operatorname{Proj} \left[ y \right]_{\mathcal{B}} \text{ if and only if } x^{*} = \Argmin_{x \in \mathcal{B}} \| y - x \|
\end{displaymath}
and $\gamma_j$ is a positive and small number called the 
step-size\footnote{Actually, there is a large variety of 
possibilities for choosing the step-size $\gamma_{j}$. See, 
e.g.~\cite{Bertsekas1999}.},  	
The convexity of $f(\cdot)$ implies that $\lim_{j \rightarrow 
\infty} f(\mu^j) = f(\mu^*)$, where $\mu^*$ is the minimum of 
$f$ constrained to $\mu^* \in \SCal^{|\ACal|}$ (see \cite{Polyak1965}). 
Equivalently, for all $\varepsilon_1 > 0$, there exists an 
$N_{\varepsilon_1} > 0$ such that
\begin{displaymath}
 | f(\mu^j) - f(\mu^*)| < \varepsilon_1
\end{displaymath}
for all $j \geq N_{\varepsilon_1}$. As mentioned above, we will approximate $\nabla f(\mu^{j-1})$ in~\eqref{eq:muClassical}. The chosen
approximation is the one used in the classical Kiefer-Wolfowitz procedure~\cite{Poznyak2009}, 
\begin{equation} \label{eq:GradientApprox}
  {Y_m(\mu^{j-1})} = \dfrac{1}{2 \beta_m} \sum_{i = 1}^{|\ACal|} \left[ f(\mu^{j-1} + \beta_m e_i) - f(\mu^{j-1} - \beta_m e_i) \right] e_i \in \mathbb{R}^{|\ACal|} \;,
\end{equation}
where the sequence $\beta_m$ vanishes as $m \to \infty$ and the vector $e_i$ represents the $i$-th element of the
canonical basis in $\mathbb{R}^{|\ACal|}$.
\begin{lemma} \label{prop:gradient}
 Consider the approximation~\eqref{eq:GradientApprox}. If the gradient $f(\cdot)$ exists at $\mu^{j-1}$ and $\lim_{m \rightarrow \infty} \beta_m = 0$,
 then the limit of $Y_m(\mu^{j-1})$ as $m \to \infty$ is equal to the gradient of $f(\cdot)$ at $\mu^{j-1}$.
\end{lemma}

\begin{proof}
The proof is in \cite{Poznyak2009} but we briefly repeat it here for completeness.
Since $f$ is differentiable at $\mu^{j-1}$, we can write the first order approximation of $f(\mu^{j-1} + \beta_m e_i)$ at $\mu^{j-1}$ as
\begin{displaymath}
 f(\mu^{j-1} + \beta_m e_i) = f (\mu^{j-1}) + \left\langle \nabla f(\mu^{j-1}), \beta_m e_i \right\rangle + o(\beta_m) \;.
\end{displaymath}
where $o(\cdot)$ satisfies
\begin{displaymath}
 \lim_{x \rightarrow 0} \bigg | \dfrac{o(x)}{x} \bigg| = 0 \;.
\end{displaymath}
Representing $f(\mu^{j-1} - \beta_m e_i)$ in the corresponding fashion gives
\begin{align*}
 Y_m(\mu^{j-1}) &= \dfrac{1}{\beta_m} \sum_{i = 1}^{|\ACal|} \left[ \langle \nabla f(\mu^{j-1}), \beta_m e_i \rangle +  o(\beta_m) \right]e_i \\
                &= \sum_{i = 1}^{|\ACal|} \left[ \dfrac{\partial f(\mu^{j-1})}{\partial \mu^{j-1}_i}  +  \dfrac{o(\beta_m)}{\beta_m} \right]e_i \\
                &= \nabla f(\mu^{j-1}) + \dfrac{o(\beta_m)}{\beta_m} \sum_{i = 1}^{|\ACal|} e_i \;.
\end{align*}
Taking the limit as $m \rightarrow \infty$ gives the desired result.
\end{proof}

With the convergence of the recursion~\eqref{eq:GradientApprox} assured, it seems natural to adjust the classical gradient
with projection algorithm as follows.

\begin{proposition}
Consider the recursion
\begin{equation} \label{eq:ProjAlgModified}
 \mu^j = \operatorname{Proj} \left[ \mu^{j-1} - \gamma_j Y_j (\mu^{j-1})\right]_{\SCal^{|\ACal|}}
\end{equation}
with $f(\cdot)$ a convex (and therefore continuous) function. Then, $\lim_{j \rightarrow \infty} f(\mu^j) = f(\mu^{*})$.
\end{proposition}

\begin{proof}
It follows from Lemma~\ref{prop:gradient} that $\nabla f(\mu^{j-1}) = \lim_{m \rightarrow \infty} Y_m(\mu^{j-1})$. By substituting
this expression in~\eqref{eq:muClassical} we obtain
\begin{multline} \label{eq:muApprox}
 \mu^j = \operatorname{Proj} \left[ \mu^{j-1} - \gamma_j \lim_{m \rightarrow \infty} Y_m (\mu^{j-1})\right]_{\SCal^{|\ACal|}}  
     \\ = \lim_{m \rightarrow \infty} \operatorname{Proj} \left[ \mu^{j-1} - \gamma_j Y_m (\mu^{j-1})\right]_{\SCal^{|\ACal|}} \;,
\end{multline}
where the second equation follows from the continuity of $\operatorname{Proj}[\cdot]_{\SCal^{|\ACal|}}$.
Let us define the term
\begin{displaymath}
 \mu^{j,m} := \operatorname{Proj} \left[ \mu^{j-1} - \gamma_j Y_m (\mu^{j-1})\right]_{\SCal^{|\ACal|}} \;.
\end{displaymath}
Then we can write~\eqref{eq:muApprox} as $\mu^j = \lim_{m \rightarrow \infty} \mu^{j,m}$ and $f(\mu^{j}) = \lim_{m \rightarrow \infty} f (\mu^{j,m})$,
where we have used again a continuity argument to compute the second limit. Thus, for all $\varepsilon_2 > 0$, there exists an $M_{\varepsilon_2} > 0$
such that
\begin{displaymath}
 | f(\mu^{j,m}) - f(\mu^j) | < \varepsilon_2 
\end{displaymath}
for all $m \geq M_{\varepsilon_2}$. Finally, for every $\varepsilon > 0$, we can choose $\varepsilon_1 > 0$ and  $\varepsilon_2 > 0$
such that $\varepsilon = \varepsilon_1 + \varepsilon_2$. It is then straightforward to verify that
\begin{align*}
 |f(\mu^{j,m}) - f(\mu^{*})| &= |f(\mu^{j,m}) - f(\mu^{*}) \pm f(\mu^j) | \\
                             &\leq | f(\mu^{j,m}) - f(\mu^j) | + |f(\mu^{j}) - f(\mu^{*}) | \\
                             &< \varepsilon_1 + \varepsilon_2 = \varepsilon  
\end{align*}
for all $j,m \geq \max \{ N_{\varepsilon_1}, M_{\varepsilon_2} \}$. 
Consequently, for the sequence $\mu^{j,j} = \mu^{j}$ we have
\begin{displaymath}
 |f(\mu^{j,j}) - f(\mu^{*})| < \varepsilon \quad \text{for all } j \geq 
  \max \{ N_{\varepsilon_1}, M_{\varepsilon_2} \} \;.
\end{displaymath}
The convergence of the recursion \eqref{eq:ProjAlgModified} has been
proved. 
\end{proof}

In order to apply the recursion~\eqref{eq:ProjAlgModified} to our 
optimization problem, it suffices to set $f(\mu) = - \mathbf{x}_0^{\top} 
\mathbf{P}_0(\mu) \mathbf{x_0}$ and specify a stop criterion. The latter can
be obtained from the complementary slackness condition~\eqref{eq:th:SlacknessCondition} 
in Theorem~\ref{th:Contribution}. More precisely: The fact that, at the optimal 
value $\mu^*$, the components $\mu_{\alpha}^*$ that are different from zero
must have equal associated costs can be used to determine whether or not an optimal $\mu^*$
has been found. In a real implementation, however, the slackness condition is impossible to
achieve exactly, due to finite machine precision. For this reason we introduce a threshold
value $\varepsilon << 1$ for which the stop condition is
\begin{equation} \label{eq:stopCondition}
 \big| \mu_{\alpha}^{j} \left[J^{\alpha}(u^*(\mu^{j-1})) - J(u^*(\mu^{j-1})) \right] \big| < \varepsilon \;,
\end{equation}
for all $\alpha \in \ACal$ and for some $j \in \mathbb{N}$ sufficiently large.

The procedure for computing the optimal control law for problem~\eqref{eq:minMaxProblem} is summarized in the following algorithm:
\begin{enumerate}
 \item Fix $\mu^{0} \in \SCal^{|\ACal|}$ and set $j = 1$. 
 \item Compute the matrix $\mathbf{P}_0(\mu^{j-1})$ as the positive-definite solution of~\eqref{eq:th:Riccati}.
 \item Make one step through the recursion~\eqref{eq:ProjAlgModified} by computing the term $\mu^{j-1} - \gamma_j Y_j (\mu^{j-1})$
  and projecting on the simplex $\SCal^{|\ACal|}$.
\item Compute the product $\mu^{j}_{\alpha} \left[ J^{\alpha}(u^*(\mu^{j})) 
- J(u^*(\mu^{j})) \right]$ for every $\alpha$ in $\ACal$ and verify 
the stop condition \eqref{eq:stopCondition}\footnote{At the expense of a large computational burden, the term
$J(u^*(\mu^{j}))$ can be obtained by integrating the system equations numerically. Alternatively,
$J(u^*(\mu^{j}))$ can be approximated by $\frac{1}{2}\mathbf{x}_0^{\top} \mathbf{P}_0(\mu^j) \mathbf{x_0}$
(notice that $\lim_{j \to \infty}\left( J(u^*(\mu^{j})) - \frac{1}{2}\mathbf{x}_0^{\top} \mathbf{P}_0(\mu^j) \mathbf{x_0} \right) = 0$).}.
If True, go to step 5;
else increase $j$ by one and go to step 2.
\item Compute the optimal control law as in~\eqref{eq:th:OptimalControl}.
\end{enumerate}

\begin{remark}
Note that, the proposed algorithm can be implemented in an MPC scenario 
by computing the $\mu_{\alpha}$ components at each prediction 
horizon. From a numerical point of view it is more efficient than 
solving the original problem.
\end{remark}
\section{Numerical examples}

\begin{example}
Consider the following multi-model system
\begin{equation} \label{eq:Ex1Sys}
 \dot{x} = A^{\alpha} x + B^{\alpha} u \;, \quad
  x_0 = \begin{bmatrix} 3 & -2 \end{bmatrix}^{\top}
\end{equation}
with $\alpha \in \ACal = \{ 1,2 \}$, 
\begin{displaymath}
A^{1} = \begin{bmatrix} 0 & 1 \\ -1 & -1 \end{bmatrix} \;, \: B^{1} = \begin{bmatrix} 0 \\ 1 \end{bmatrix} \;, \:
A^{2} = 10 A^{1} \;, \: B^{2} = B^{1} \;.
\end{displaymath} 
Also, consider the cost functional
\begin{equation} \label{eq:Ex1Cost}
\begin{split}
 J^{\alpha}(u) &= \dfrac{1}{2} x^{\alpha \top}(10) G x^{\alpha}(10) + \int_{0}^{10} 
  \left( x^{\alpha \top} Q x^{\alpha} + u^{\top} R u \right) \mathrm{d}t \\
    G &= \begin{bmatrix} 5 & 0 \\ 0 & 5 \end{bmatrix} \;, \quad Q = \begin{bmatrix} 50 & 0 \\ 0 & 10 \end{bmatrix} \;, \quad R = 10 \;.
\end{split}
\end{equation}
Note that the worst dominant plant (i.e., the plant whose individual cost $J^{\alpha_w}$ 
is always greater than the other costs $J^{\alpha}$, for all 
$\alpha \in \ACal \setminus \{ \alpha_w \}$ and for all admissible 
controls) is given by $\alpha = 1$ (the slowest plant). We expected the proposed
algorithm to be able to identify it.
In this case we consider a random switching sequence
\begin{align*}
\delta &= \delta_1 \cup \delta_2 \\
 \delta_1 &= \begin{bmatrix} 0 & 0.82 & 1.73 & 1.86 & 2.78 & 3.42 & 3.52 & 3.80 & 4.35 \end{bmatrix} \\ 
 \delta_2 &= \begin{bmatrix}  5.31  &  6.28 &  6.44  &  7.42  &  8.38  &  8.87  &  9.68 & 9.83 & 10 \end{bmatrix}. 
\end{align*}

Setting $\varepsilon = 2 \times 10^{-5}$ and applying the algorithm described in Section~\ref{sec:numAlg}, one obtains the optimal vector
\begin{displaymath}
 \mu^{*} =  \begin{bmatrix} 1 & 0 \end{bmatrix}^{\top}.
\end{displaymath}
After substituting the optimal parameter in~\eqref{eq:th:OptimalControl}, the optimal control $u^{*}$
is obtained. Such control is plotted in Fig.~\ref{fig:Ctrl_Ex1}.
The optimal min-max cost $J(u^{**})$ is equal to $139.1381$. The individual costs that
result from applying the optimal min-max control $u^{**}(\cdot)$ to every $\alpha$-system independently
are $J^{1}(u^{**}) = 139.1381$ and $J^{2}(u^{**}) = 20.7546$, which confirms that the worst plant was
identified correctly. Finally, the state trajectories for each system are presented in
Figs.~\ref{fig:Sys1_Ex1} and~\ref{fig:Sys2_Ex1}.
\begin{figure}[h]
\begin{center}
 \includegraphics[width=0.45\textwidth]{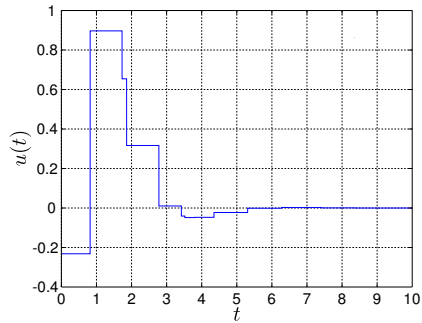}
 \caption{Optimal min-max control.}
 \label{fig:Ctrl_Ex1}
\end{center}
\end{figure}
\begin{figure} 
\centering
\subfloat[$ \alpha = 1 $]{\label{fig:Sys1_Ex1}\includegraphics[width=0.4\textwidth]{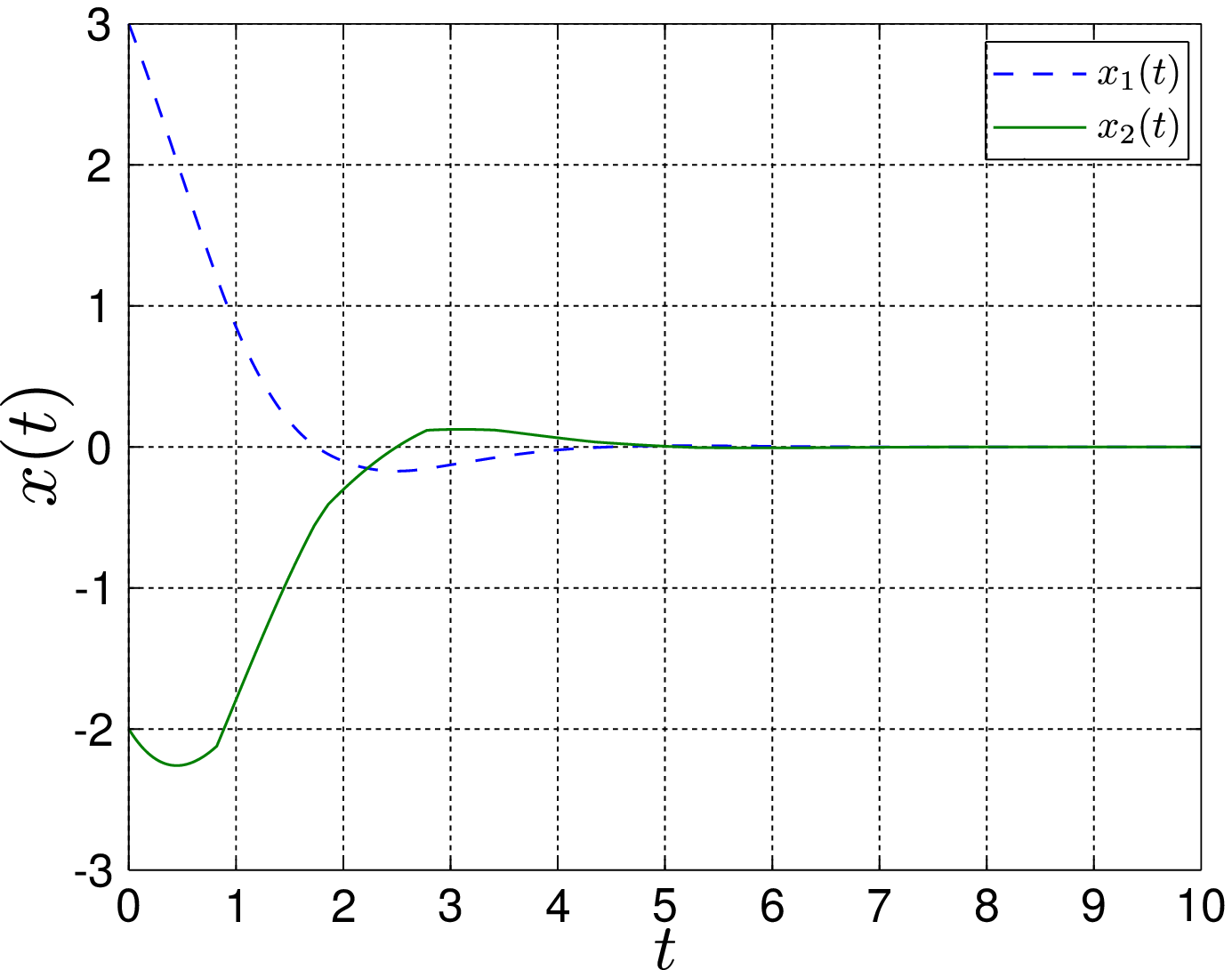}} \;
\subfloat[$ \alpha = 2 $]{\label{fig:Sys2_Ex1}\includegraphics[width=0.4\textwidth]{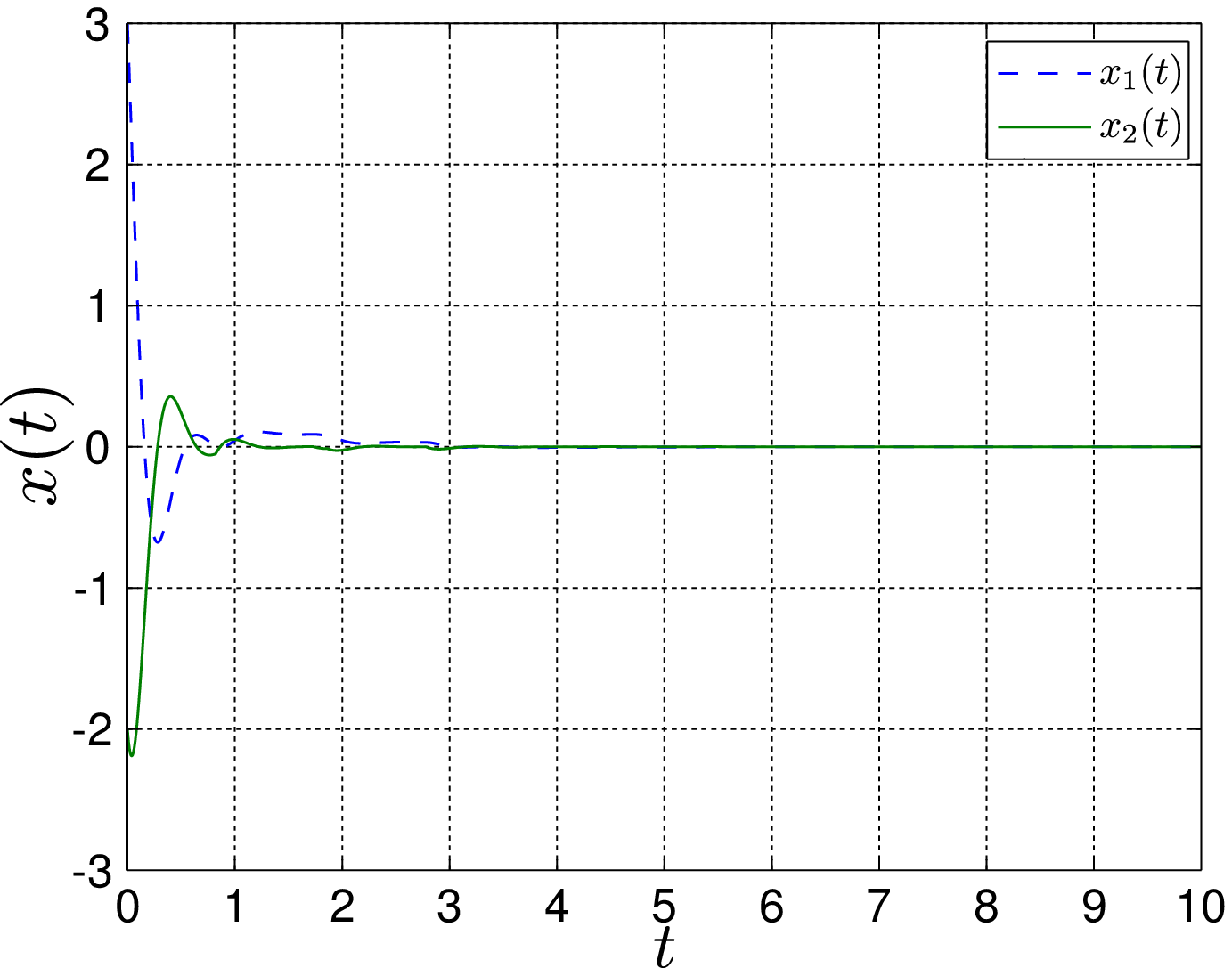}} \\
\caption{Time trajectories for the system \eqref{eq:Ex1Sys} with $\alpha = 1,2$} \label{fig:Sys1-2}
\end{figure}
\end{example}
\begin{example}
Consider now the linear multi-model system
\begin{equation} \label{eq:Ex2Sys}
\begin{split}
       \dot{x}^{\alpha} &= A^{\alpha} x^{\alpha} + B^{\alpha} u \\
             A^{\alpha} &= \begin{bmatrix} 0 & 1 \\ (\alpha-0.9) \Sign(1.1 - \alpha) & -(4 -\alpha)^{2} \end{bmatrix} \;, \quad 
              B^{\alpha} = \begin{bmatrix} 0 \\ \sqrt{\alpha} \end{bmatrix} \;, \quad
              x_0 = \begin{bmatrix} -5 \\ 3 \end{bmatrix} \\
\end{split}
\end{equation}
with $\alpha \in \ACal = \{ 1,2,3,4 \}$ and the cost functional
\begin{equation} \label{eq:Ex2Cost}
\begin{split}
 J(u) &= \dfrac{1}{2} x^{\alpha \top} G x^{\alpha} + \int_{0}^{20} \left( x^{\alpha \top} Q x^{\alpha} + u^{\top} R u \right) \mathrm{d}t \\
    G &= \begin{bmatrix} 5 & 0 \\ 0 & 5 \end{bmatrix} \;, \quad Q = \begin{bmatrix} 50 & 0 \\ 0 & 10 \end{bmatrix} \;, \quad R = 10 \;.
\end{split}
\end{equation}
In this case we consider a random switching sequence $\delta$ given by:
\begin{align*}
 \delta_1 &= \begin{bmatrix} 0 & 0.82 & 1.73 & 1.86 & 2.78 & 3.42 & 3.52 & 3.80 & 4.35\end{bmatrix} \\ 
 \delta_2 &= \begin{bmatrix}  5.31  &  6.28 &  6.44  &  7.42  &  8.38  &  8.87  &  9.68 & 9.83  & 10.26 \end{bmatrix} \\ 
 \delta_3 &= \begin{bmatrix} 11.18 &  11.98  & 12.94 &  13.60 &  13.64 &  14.49 & 15.43 & 16.11 & 16.87 \end{bmatrix}  \\ 
 \delta_4 &= \begin{bmatrix} 17.62 & 18.02 &  18.68 & 20.00 & 20.52 & 22.36  & 23.96 &  25.88 &   27.20 \end{bmatrix} \\
  \delta_5 &= \begin{bmatrix}  27.28 & 28.98 &  30.86 &  32.22 & 33.74 & 35.24 &  36.04 &  37.36 &  40.00 \end{bmatrix}\\
   \delta &= \cup_{i = 1, \ldots 5} \delta_i.
\end{align*}

Applying the algorithm described with $\varepsilon = 5 \times 10^{-3}$ gives
\begin{displaymath}
 \mu^{*} =  \begin{bmatrix} 0.4842 & 0.1842 & 0.1432 & 0.1884 \end{bmatrix}^{\top}
\end{displaymath}
(notice that the model is unstable and marginally stable for $\alpha = 1$ and $\alpha = 4$, respectively).
Fig.~\ref{fig:Ctrl_Ex2} shows the optimal control~\eqref{eq:th:OptimalControl} after plugging in $\mu^*$.
The optimal min-max cost is equal to $J(u^{**}) = 3688.1$ and the individual costs are
\begin{equation} \label{eq:Jss}
 J^{1}(u^{**}) = 3688.1 , \; J^{2}(u^{**}) = 3688.1, \; J^{3}(u^{**}) = 3688.1  \; \text{and} \; J^{4}(u^{**}) = 3688.1 \;.
\end{equation}
In this case, all the plants play a role in the optimal control
and can be viewed as extreme plants in the min-max sense. 
The min-max control strategy is well suited for the multi-model case in 
the sense that, \emph{for every plant}, an appropriate upper bound on the cost is ensured.
For comparison purposes, we have computed each control $u^{\alpha *}$, 
where $u^{\alpha *}$ is defined as the optimal control for the plant $\alpha$.
Next, we have computed the cost that each plant incurs when subject to
$u^{\alpha *}$. All costs are shown in Table \ref{table:Indiv:costs}. It can be seen
that, for all $u^{\alpha *}$, there is at least one plant that incurs a cost which is larger than $J(u^{**})$. This
supports the claim that the min-max is a reasonable criterion.
Finally, the state trajectories for each system are presented in Fig.~\ref{fig:sys_Ex2}.
\begin{table}
\centering
\begin{tabular}{c|cccc} 
 				& $J^{1}(u)$ & $J^{2}(u)$ & $J^{3}(u)$ 				& $J^{4}(u)$      \\ \hline
 $u = u^{1 *}$  & $2384.4$ & $4900.0$ & $7649.7$ 	& $1.22 \times 10^{5}$ \\
 $u = u^{2 *}$  & $2.462 \times 10^{4}$ & $570.77$ & $1526.7$ 	& $6.465 \times 10^{4}$ \\
 $u = u^{3 *}$  & $3.889 \times 10^{4}$ & $1194.2$ & $381.16$ 				& $1.269 \times 10^{4}$ \\
 $u = u^{4 *}$  & $4.454 \times 10^{4}$ & $1749.6$ & $691.35$ 	& $485.76$
\end{tabular}
\caption{The costs that each plant incurs when subject to $u^{\alpha *}$, the optimal control for the $\alpha$-system. Compare
 with~\eqref{eq:Jss}.}
\label{table:Indiv:costs}
\end{table}
\begin{figure}
\begin{center}
 \includegraphics[width=0.45\textwidth]{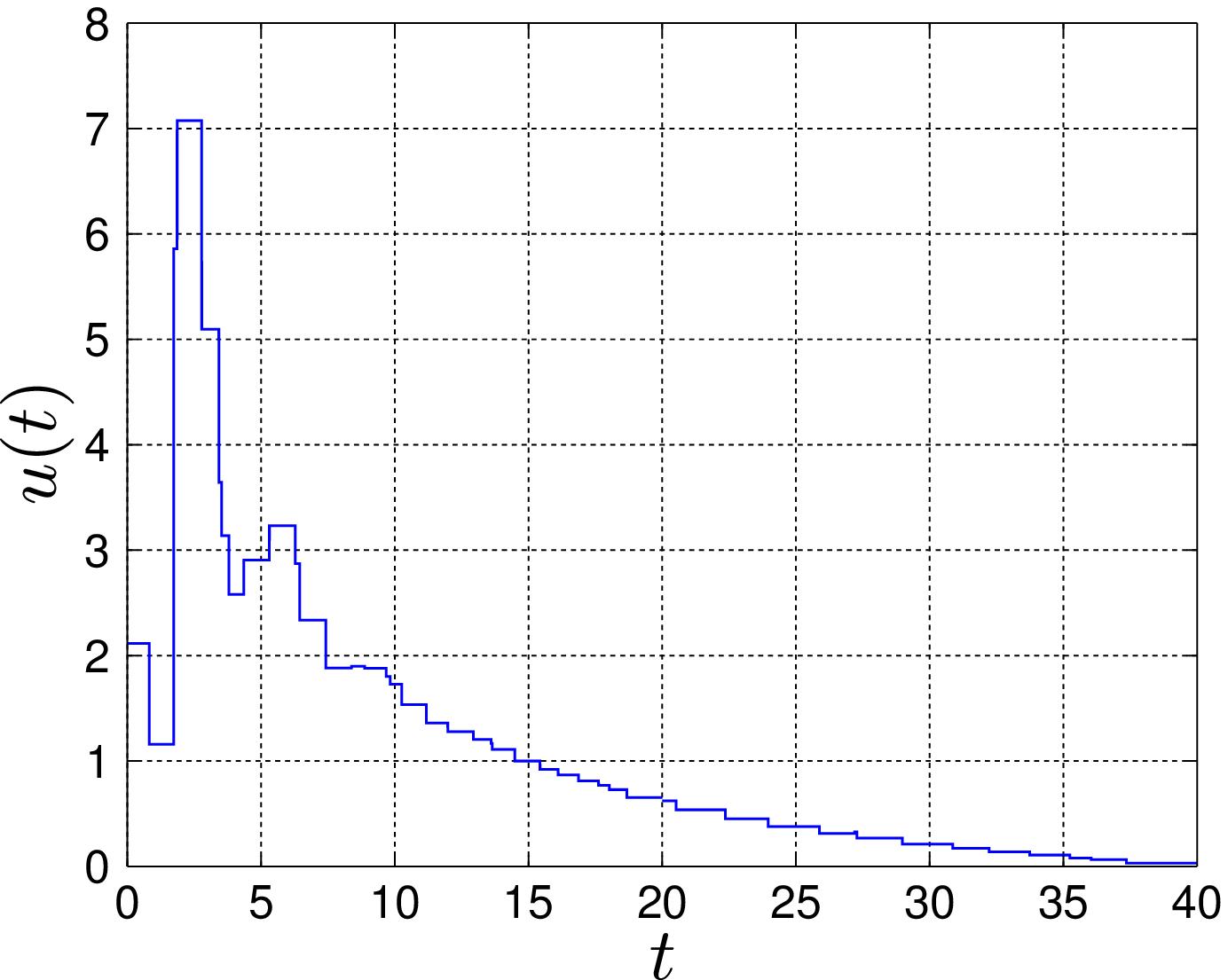}
 \caption{Optimal min-max control, $u^{**}$.}
 \label{fig:Ctrl_Ex2}
\end{center}
\end{figure}
\begin{figure}
\centering
\subfloat[$\quad \alpha = 1$.]{\label{fig:Sys1_Ex2}\includegraphics[width=0.45\textwidth]{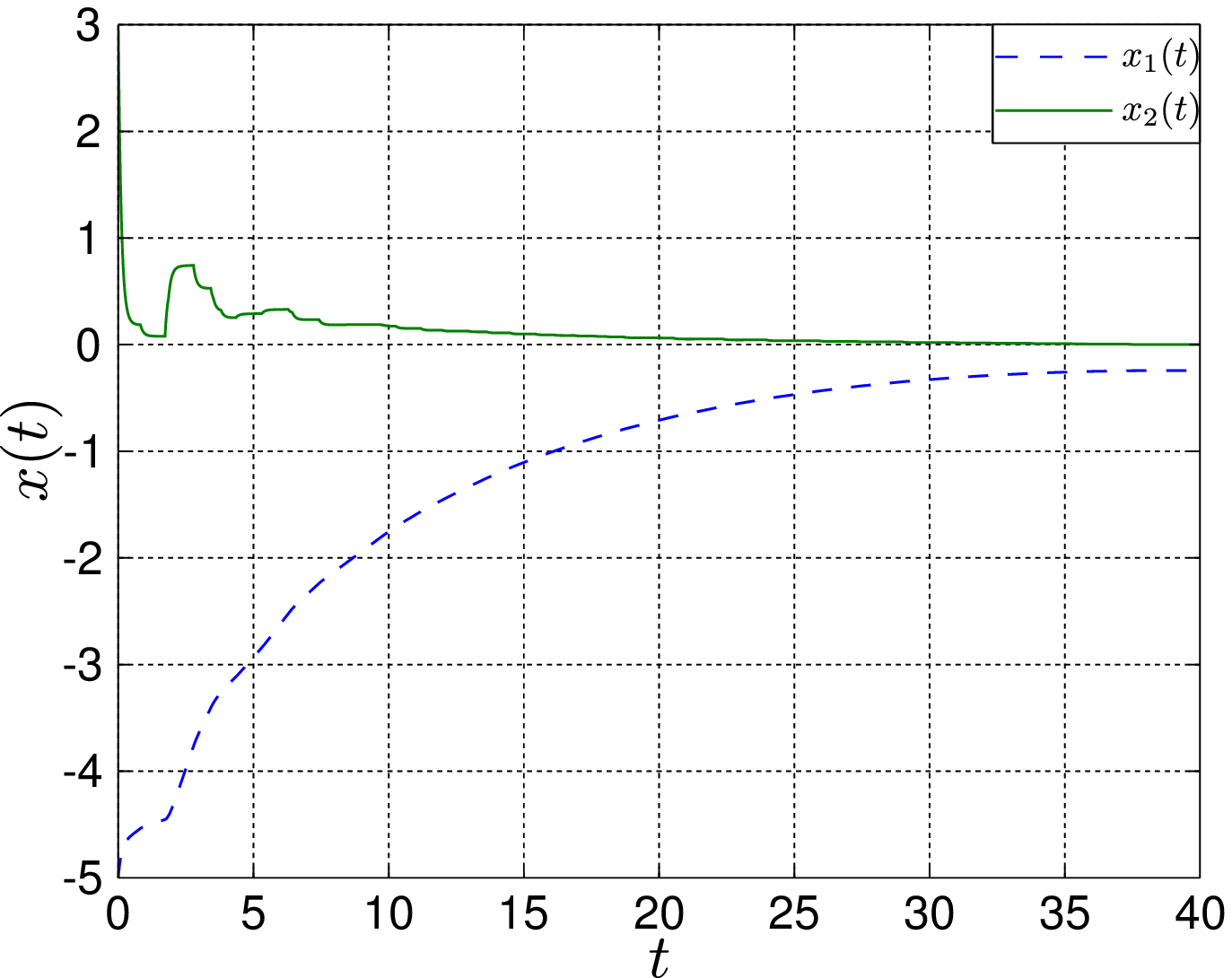}} \; 
\subfloat[$\quad \alpha = 2$.]{\label{fig:Sys2_Ex2}\includegraphics[width=0.45\textwidth]{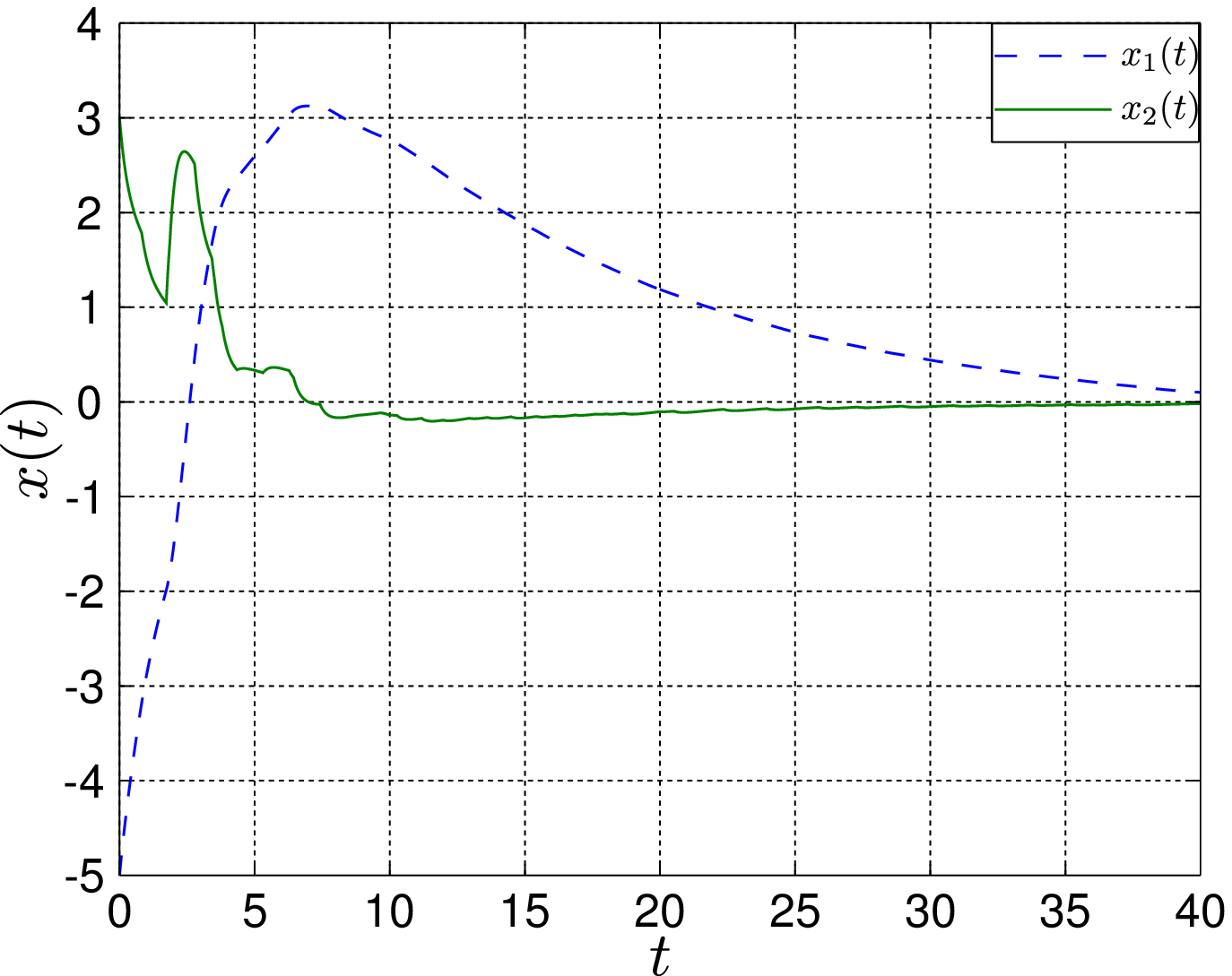}} \\
\subfloat[$\quad \alpha = 3$.]{\label{fig:Sys3_Ex2}\includegraphics[width=0.45\textwidth]{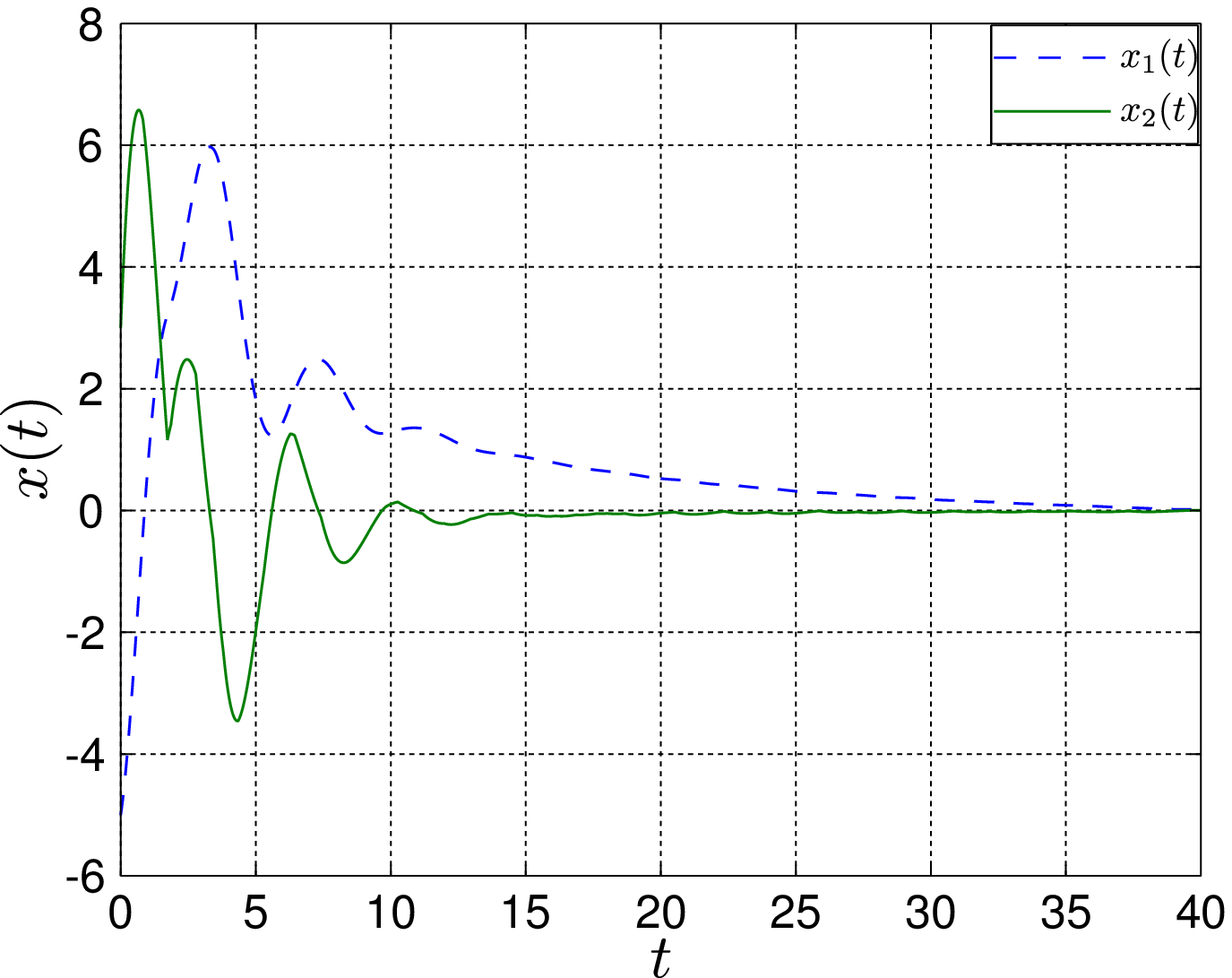}} \;  
\subfloat[$\quad \alpha = 4$.]{\label{fig:Sys4_Ex2}\includegraphics[width=0.45\textwidth]{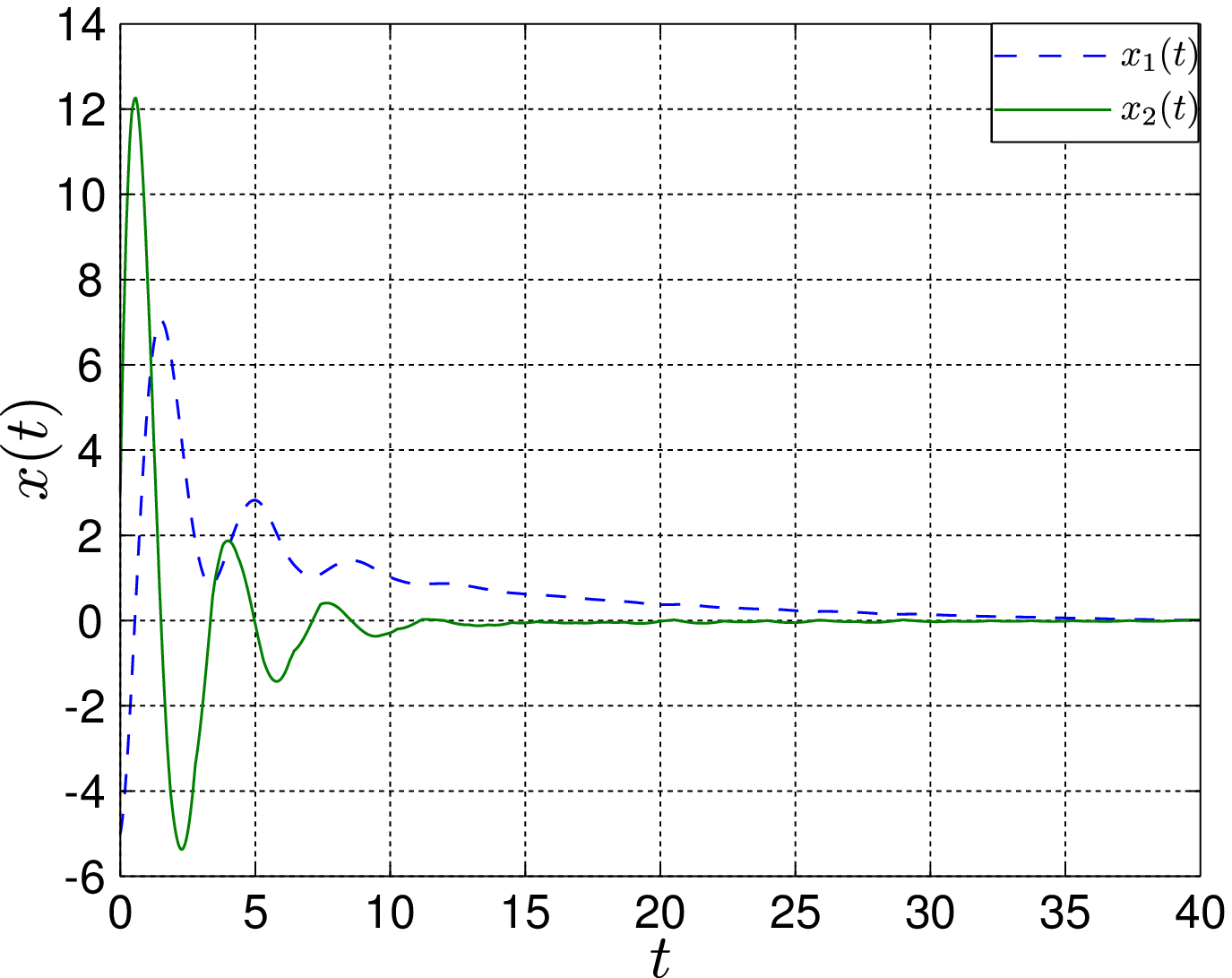}}
\caption{Trajectories of~\eqref{eq:Ex2Sys} subject to $u^{**}$.}\label{fig:Ex2States}
\label{fig:sys_Ex2}
\end{figure}

\end{example}
\section{Conclusions}

By employing generalized gradients we have formulated a multi-model method of Lagrange multipliers. When
applied to the discrete-time min-max optimal control problem, the method leads to a Riccati equation
for an extended plant with a state vector obtained by aggregating the state of each individual plant.
This is in perfect analogy with the continuous-time solution~\cite{Boltyanski2011}. The Riccati equation
is parametrized by $\mu$, a member of a simplex whose elements determine the weight assigned to the cost
of each plant. Non-smooth analysis specifies $\mu$ as the solution of a maximization problem over a simplex,
thus completely characterizing the optimal solution. Non-smooth analysis also leads to a complementary
slackness condition on $\mu$, which turns out to be useful when computing the solution numerically.

Numerical experiments show the effectiveness of the min-max approach in the context of a multi-model
setting, in the sense that the cost of each plant is kept at a reasonable level. It is worth mentioning,
however, that the resulting optimal control is essentially open-loop, since it depends on the
state-trajectories of all the models, whether they are actually realized or not. The problem of obtaining
a state-feedback control thus remains open, but points the direction for continuing this line of research.



\section*{Appendix}
\begin{lemma} \label{lemma:EqProblems}
Consider the problem of finding the maximum element 
among a finite set indexed by $\ACal$, i.e.,
$\max_{\alpha \in \ACal} \{ z^{1}, \ldots, z^{|\ACal|} \}$. 
This problem is equivalent to the following linear program:
\begin{displaymath}
 \operatorname{maximize}_{\mu \in \SCal^{|\ACal|}} 
  \sum_{\alpha \in \ACal} \mu_{\alpha} z^{\alpha} \;.
\end{displaymath}
Moreover, the solution $\mu^{*}$ of the linear program satisfies $\mu_k^{*} = 0$ 
for all $k \notin I := \{\alpha \in \ACal : z^{\alpha} = z^{0}\}$ with 
$z^0 := \max \{z^{1}, \ldots, z^{|\ACal|}\}$.
\end{lemma}
\begin{proof}
Notice that
\begin{align*}
 z^{0} = \sum_{\alpha \in \ACal} \mu_{\alpha} z^{0} \geq 
  \sum_{\alpha \in \ACal} \mu_{\alpha} z^{\alpha} \quad \text{for all } \mu \in \SCal^{|\ACal|} \;,
\end{align*}
so that $z^{0} \geq \zeta^{*} := \max_{\mu \in \SCal^{|\ACal|}} \sum_{\alpha \in \ACal} \mu_{\alpha} z^{\alpha}$.
On the other hand, we have $\zeta^{*} \geq z^{\alpha}$ for all $\alpha \in \ACal$. Therefore,
$z^{0} = \zeta^{*}$. The proof is established by contradiction: suppose that there exists some indices
$b \in \bar{I}: = \{ \alpha \in \ACal : z^{\alpha} < z^{0} \}$ such that $\mu_{b}^{*} \neq 0$. We have
\begin{align*}
 z^{0} &= \sum_{\alpha \in \ACal} \mu_{\alpha}^{*} z^{\alpha} \\
       &= \sum_{\alpha \in I} \mu_{\alpha}^{*} z^{0} + \sum_{\alpha \in \bar{I}} \mu_{\alpha}^{*} z^{\alpha} \\
       &\leq \sum_{\alpha \in I} \mu_{\alpha}^{*} z^{0} + \sum_{\alpha \in \bar{I}} \mu_{\alpha}^{*} \tilde{z} \;,
\end{align*}
where $\tilde{z} = \max_{\alpha \in \bar{I}} z^{\alpha}$. The last inequality implies that
\begin{displaymath}
 \sum_{\alpha \in \bar{I}} \mu_{\alpha}^{*} z^{0} = \left( 1 - \sum_{\alpha \in I} \mu_{\alpha}^{*} \right) z^{0} \leq \sum_{\alpha \in \bar{I}} \mu_{\alpha}^{*} \tilde{z} \;.
\end{displaymath}
In other words, $z^{0} \leq \tilde{z}$, the desired contradiction. 
\end{proof}

\bibliographystyle{wileyj}
\bibliography{biblio}
\end{document}